%% file: arXiv_2011.tex
\newcommand{\ver}{arXiv 4.6}

\documentclass[11pt]{article}
\usepackage{rax_defs}
\usepackage{clrscode3e}
\usepackage[english]{babel}

\usepackage{fullpage}

\usepackage[compact]{titlesec}

\newtheorem{lemma}{Lemma}[section]
\newtheorem{theorem}{Theorem}[section]

\newtheorem{assumption}{Assumption}[section]
\numberwithin{equation}{section}
\numberwithin{figure}{section}

\newcommand{\ieor}{Department of Industrial Engineering \& Operations Research, Columbia University}

\renewcommand{\seqs}{\vspace{-.4cm}\begin{eqnarray}}
\renewcommand{\seqe}{\end{eqnarray}\vspace{-.4cm}}

\newcommand{\spd}{s}
\newcommand{\mspd}{\sigma}

\newcommand{\p}[1]
{
	\ifthenelse{\isempty{#1}}{\Pi}{\pi(#1)}
}

\newcommand{\xb}{\bar{x}}
\newcommand{\xt}{\tilde{x}}

\newcommand{\pw}{\beta}
\newcommand{\ect}[1]{$1|#1|\sum E_i(s_i)+\sum w_iC_i$}
\newcommand{\etr}[1]{$1|#1|\sum E_i(s_i)+\sum w_iT_i$}
\newcommand{\gect}[1]{$1|#1|\sum \mathcal{E}_i(s_i)+\sum w_iC_i$}
\newcommand{\getr}[1]{$1|#1|\sum \mathcal{E}_i(s_i)+\sum w_iT_i$}

\newcommand{\algI}{SAIAS}
\newcommand{\algIlong}{Schedule by $\alpha$-intervals and $\alpha$-speeds}
\newcommand{\algdue}{SAIAS-T}
\newcommand{\algduelong}{Schedule by $\alpha$-intervals and $\alpha$-speeds for Tardiness}

\input{rax_hsetup}
\newhead{Energy Aware Scheduling v.\ver}{Carrasco, Iyengar, Stein}

\title{Energy Aware Scheduling for Weighted Completion Time and Weighted Tardiness}

\author{
	Rodrigo A. Carrasco
	\thanks{ {\tt rac2159@columbia.edu}.
		\ieor, Mudd 313, 500W 120th Street, New York, NY 10027. Research partially supported by NSF grants CCF-0728733 and CCF-0915681, and Fulbright/Conicyt Chile Scholarship.
	}
	\and
	Garud Iyengar
	\thanks{ {\tt garud@ieor.columbia.edu}.
		\ieor, Mudd 314, 500W 120th Street, New York, NY 10027. Research partially supported by NSF grant DMS-1016571, ONR grant N000140310514, and DOE grant DE-FG02-08ER25856.
	}
	\and
	Cliff Stein
	\thanks{ {\tt cliff@ieor.columbia.edu}.
		\ieor, Mudd 326,500W 120th Street, New York, NY 10027. Research partially supported by NSF grants CCF-0728733 and CCF-0915681.
	}
}

\date{\monthname ~2011, v.\ver}

\begin{document}
\maketitle

\thispagestyle{empty}
\begin{abstract}
	The ever increasing adoption of mobile devices with limited energy storage capacity, on the one hand, and more awareness of the environmental impact of massive data centres and server pools, on the other hand, have both led to an increased interest in energy management algorithms.
	
	The main contribution of this paper is to present several new constant factor approximation algorithms for energy aware scheduling problems where the objective is to minimize weighted completion time plus the cost of the energy consumed, in the one machine non-preemptive setting, while allowing release dates and deadlines.Unlike previous known algorithms these new algorithms can handle general job-dependent energy cost functions, extending the application of these algorithms to settings outside the typical CPU-energy one. These new settings include problems where in addition, or instead, of energy costs we also have maintenance costs, wear and tear, replacement costs, etc., which in general depend on the speed at which the machine runs but also depend on the types of jobs processed. Our algorithms also extend to approximating weighted tardiness plus energy cost, an inherently more difficult problem that has not been addressed in the literature.
\end{abstract}
\newpage


\section{Introduction}
\label{sec:intro}
Managing energy consumption is a problem of critical interest throughout the world and throughout various industries. Computing devices use a large amount of energy, both in individual devices such as laptops and PDAs and also in large industrial uses such as datacenters. For example, Google states that the servers in its datacenter, which are much more efficient than the average industry server, consume  $1$kJ per query on average \cite{Google09}. In January 2011, just in the US, there were an average more than $400$ million queries per day \cite{Comscore11}, and thus the total amount of energy consumed was $44.5$ million kWh, equivalent to more than $4,000$ average US households \cite{DOE11}. Furthermore, CPUs account for 50-60\% of a typical computer's energy consumption \cite{Albers2009}, making CPU energy management very important. When scheduling on such devices, it is important not only to consider the relevant quality of service (QoS) metrics such as makespan or weighted completion time, but also to take energy consumption into account.
Most modern CPUs can be run at multiple speeds; the lower the speed, the less energy used, and the relationship is device-dependent, but typically superlinear.
The technique of scheduling while controlling the speed of the processor is known as {\em speed scaling.}

Starting with the work of Yao, Demers, and Shenker \cite{Yao1995}, there has by now been tens of papers studying scheduling problems in which energy consumption is taken into account. (See, for example, the surveys by Irani and Pruhs \cite{Irani2005} and that by Albers \cite{Albers2009}). There are three main settings for energy aware scheduling problem: optimizing a QoS metric with an energy budget \cite{Pruhs2007, Pruhs2008}, minimizing energy subject to a QoS constraint \cite{Bansal2008, Bansal2004, Bansal2007a, Yao1995}, or optimizing some convex combination of a scheduling objective and energy consumption \cite{Albers2007, Andrew2009, Bansal2009, Bansal2007b}. Underlying the latter setting, which is in the one we will focus in this work, is an assumption that both energy and time can be (implicitly) converted into a common unit, such as dollars.

\subsection{Our results}
In this paper we consider two commonly studied scheduling metrics, {\em  weighted completion time} and {\em weighted tardiness}, that have not received attention in the energy aware scheduling literature. Given a schedule in which job $i$ with  weight $w_i$, release time $r_i$, and deadline $d_i$ is completed at time $C_i$, the total weighted completion time is $\sum_i w_i C_i$. The tardiness of a job is zero if it is completed before its deadline and otherwise equal to the amount by which it misses, that is, $T_i = \max\{0, C_i-d_i\}$ and total weighted tardiness is $\sum w_iT_i$. For both these metrics, we consider the non-preemptive, off-line problem on one machine, and allow arbitrary precedence constraints. For the weighted completion time we allow arbitrary release dates as well. We consider a metric that is a convex combination of our scheduling metric and energy cost. We are not aware of any previous work on energy aware scheduling algorithms for these metrics. There is a rich literature on minimizing weighted completion time in the absence of energy concerns (e.g. \cite{Phillips1998, Pinedo2008,Skutella2006}), but we are aware of only one result about weighted tardiness in the absence of energy concerns in the speed scaling/resource augmentation literature \cite{Bansal2007c}, where a $2$-machine, $24$-speed $4$-approximation algorithm is presented. Weighted tardiness, in particular, is difficult to analyze because, in contrast to most scheduling objectives, it is a non-linear function of completion time.

In our work we consider a more general model of energy cost than has previously been used. The most common energy model assumes that the rate at which power is consumed is a polynomial function of speed of the form $P(s)=s^{\pw}$ for some constant $\pw$; typical values of $\beta$ are 2 or 3. Some recent work\cite{Andrew2009, Bansal2009} uses a more general power function with minimum regularity conditions, like non-negativity, but in all the cases the power function does not depend on the job. Furthermore, most energy aware algorithms assume cost functions that are closely related to energy consumption; however, in practice the actual energy cost is not simply a function of energy consumption, it is a complicated function of discounts, pricing, time of consumption, etc. We consider a more general class of cost functions that are only restricted to be non-negative and can be different for different jobs. Because we allow job-dependent energy costs, our algorithms can be used outside the CPU-energy setting, where energy cost generally are job independent, and can be applied to more general problems that have additional speed-associated costs. Examples of these costs are maintenance costs, wear and tear of parts, failure rates, etc. all of which not only depend on the speed at which the machine runs, but also the job being processed. We are not aware of any other work that allows such general costs. For the weighted tardiness case we require an additional regularity condition on the energy cost functions that allows us to control its rate of growth.

Our paper contains several results for different scheduling problems, we state here the most general results:
\begin{theorem}
Given $n$ jobs with precedence constraints and release dates and a general non-negative energy cost function, there is an $O(1)$-approximation algorithm for the
problem of  non-preemptively minimizing a convex combination of weighted completion time and energy cost.
\label{th:general_wC}
\end{theorem}

\begin{theorem}
Given $n$ jobs with precedence constraints and deadlines and a general non-negative energy cost function, there is an $O(1)$-approximation algorithm for the
problem of non-preemptively minimizing a convex combination of weighted tardiness and energy cost.
\label{th:general_wT}
\end{theorem}

The constants in the $O(1)$ are modest. Consider the case where we are given a set of speeds $\vk{S}=\set{\mspd_1,\ldots,\mspd_m}$, at which the machine can run, with $\mspd_j\leq(1+\delta)\mspd_{j-1}$, and some $\epsilon>0$. Then the algorithm for the weighted completion time setting has a $4(1+\epsilon)(1+\delta)$-approximation ratio when only precedence constraints exist, and $(3+2\sqrt{2})(1+\epsilon)(1+\delta)$-approximation ratio when release dates are added.  The algorithm for the weighted tardiness setting has a $4^\pw(1+\epsilon)^{\pw-1}(1+\delta)^{\pw-1}$-approximation ratio even with arbitrary precedence constraints, where $\pw$ controls the growth of the energy cost function.

\subsection{Our Methodology}
The problem of minimizing weighted completion time in the combinatorial setting has been well-studied. The work of Phillips, Stein, and Wein \cite{Phillips1998} and Hall, Schulz, Shmoys, and Wein \cite{Hall1997, Hall1997a} introduced the idea of {\em $\alpha$-points}, and these  have been used in much of the subsequent work.
The idea is that one first formulates a time-indexed integer program in which decision variable $x_{it}$ is 1 if job $i$ completes at time $t$, and then solves its linear programming relaxation. From the solution to the relaxation, one computes the $\alpha$-point of each job, that is, the earliest time at which an $\alpha$ fraction of the job has completed in the relaxation. The exact interpretation of when an $\alpha$ fraction completes depends upon the particular problem. One uses these $\alpha$-points to infer an order on the jobs and then runs the jobs non-preemptively, respecting that order. There are many variants and extensions of these technique including choosing $\alpha$ randomly \cite{Chekuri2001, Goemans1997} or choosing a different $\alpha$ for each job \cite{Goemans2002}.  This technique has led to small constant factor approximation algorithms for many weighted completion time scheduling problems \cite{Skutella2006}.

The time-indexed integer program (IP) formulations for this problem are not typically of polynomial size. However, the {\em interval-indexed} IP, introduced in \cite{Hall1997a}, in which time is divided into geometrically increasing intervals and jobs are assigned to intervals rather than individual time slots, is of polynomial size. By using this linear program one obtains a polynomial sized linear program from which it is still possible to apply the ideas of $\alpha$-points while suffering only a small additional degradation of the approximation ratio.

In this paper, we extend the interval-indexed IP to handle speed scaling and then design new $\alpha$-point based rounding algorithms to obtain the resulting schedules. In doing so we introduce the new concept of {\em $\alpha$-speeds}. We assume, in Sections \ref{sec:formulation}, \ref{sec:approx_alg}, and \ref{sec:tardiness}, that we have a discrete set of $m$ allowable speeds $\vk{S} = \set{\mspd_j}$, and that the rate of power consumption is a polynomial function of the speed. In Section \ref{sec:extension_cspeeds} we describe how to remove these assumptions. Although the time-indexed IPs are easier to explain, due to limited space, we will describe only the interval-indexed linear programs in this paper. In our interval-indexed IP, a variable $x_{ijt}$ is 1 if job $i$ runs at speed $\mspd_j$ and completes in interval $t$. We can then extend the standard interval-indexed integer programming formulation to take the extra dimension of speed into account (see Section \ref{sec:formulation} for details). Once we have solved its linear program (LP) relaxation, we need to now determine {\em both an $\alpha$-point and $\alpha$-speed}. The key insight is that by ``summarizing" each dimension appropriately, we are able to make the correct choice for the other dimension. At a high level, we first choose the $\alpha$-point by ``collapsing" all pieces of a job that complete in the LP in interval $t$ (these pieces have different speeds), being especially careful with the last interval, where we may have to choose only some of the speeds. We then use {\em only} the pieces of the job that complete before the $\alpha$-point to choose the speed, where the speed is chosen by collapsing the time dimension and  then interpreting the result as a probability mass function (pmf), where the probability that the job is run at speed $\mspd_j$ depends on  the total amount of processing done at that speed. We then define the concept of $\alpha$-speeds, which is related to the expected value under this pmf, and run the job at this speed (see Section \ref{sec:approx_alg} for more details). We combine this new rounding method with extensions of the more traditional methods for dealing with precedence constraints and release dates to obtain our algorithms.

For weighted tardiness, we emphasize again that not much is known about approximating this problem, even in the absence of energy concerns. For this problem, we are able to use the same interval-indexed linear program, with the objective function modified to tardiness. Because the linear program is interval indexed, the non-linear objective function is not a problem. After the solving the linear program, we are able to show that with only a constant factor increase in energy (over the lower bound from the linear program), we obtain only a constant factor (over the linear program) increase in tardiness. Implicit in this analysis is the fact that jobs that receive $0$ tardiness in the linear program will receive $0$ tardiness in our solution; in some sense the speed scaling makes accomplishing this easier than in the combinatorial setting. We note that our weighted tardiness algorithms does not work in the presence of release dates, as release dates may stop us from being able to keep jobs with $0$ tardiness in the LP at $0$ tardiness in the  schedule.

Finally, in Section~\ref{sec:extension}, we show how to extend our results for the weighted completion and weighted tardiness scheduling metrics to general energy cost functions. We also show how to extend our results to the setting where continuous speeds are used and not just a discrete set $\vk{S}$, while maintaining the same approximation ratio.

\section{Problem Formulation}
\label{sec:formulation}

\subsection{Problem Setting}
We are given $n$ jobs, where job $i$ has a processing requirement of $\rho_i\in\mathbb{N}_{+}$ machine cycles, release time $r_i$, and an associated positive weight $w_i$. Let $\spd_i$ denote the speed at which job $i$ runs on the machine and $C_i$ denote its completion time. Let $\p{} = \{\p{1}, \ldots,\p{n}\}$ denote the order in which the jobs are processed, i.e. $\p{k} = i$ implies that job $i$ is the $k$-th job to be processed. Then $C_{\p{i}} = \max\{r_{\p{i}}, C_{\p{i-1}}\} + \frac{\rho_{\p{i}}}{\spd_{\p{i}}}$ is the completion time of the $i$-th job to be processed, with $C_{\p{0}} = 0$. We do not allow preemption.

Let $\vk{S}=\set{\mspd_1, \ldots, \mspd_m}$, be the set of possible speeds at which the machine can run. We will assume that $\mspd_{j+1} \leq (1+\delta)\mspd_{j}$, for some $\delta>0$. This is a natural assumption because actual speed scaling achieved in CPUs is done via frequency multipliers or dividers. Although a discrete set of speeds is probably the most common case for CPUs, in Section \ref{sec:extension_cspeeds} we show that our algorithm has the same approximation ratio when a continuous set of speeds is used.

Let $E_i(\spd_i)$ denote the energy cost of running job $i$ at speed $\spd_i$. For simplicity we initially consider $E_i(\spd_i) = v_i\rho_i\spd_i^{\pw-1}$, where $\pw\geq 2$ and $v_i$ are known constants. As indicated earlier, an energy cost function of this form is the standard model for these problems, although our model is more general because the energy cost function is job-dependent. In Section \ref{sec:extension} we show that our algorithms also work for a much larger class of job-dependent energy cost functions.

The objective is to compute a feasible schedule $(\p{}, \vk{C})$, consisting of an order $\p{}$ and completion times $\vk{C}$, possibly subject to \textit{}precedence and/or release date constraints, and the vector of job speeds $\vk{\spd}=\{\spd_1,\ldots,\spd_n\}\in\mathbb{R}_{+}^n$ that minimizes the total cost,
\begin{eqnarray}
f(\p{},\vk{\spd}) = \sum_{i=1}^n \left[v_i \rho_i \spd_i^{\pw-1} + w_{\p{i}} C_{\p{i}}\right],
\label{eq:general_costfunction}
\end{eqnarray}

Since this function is convex we can assume, w.l.o.g., that each job runs at a constant speed.

For convenience we will use an extended version of the notation of Graham et al. \cite{Graham1979} to refer to the different energy aware scheduling problems, i.e. \ect{r_i,prec}, will refer to the problem setting with $1$ machine, with $r_i$ release dates, precedence constraints, and the weighted completion time as the scheduling performance metric. Similarly, the \etr{r_i,prec} will refer to the same setting, but with tardiness as the scheduling performance metric. In all of them $E_i(\spd_i)$ indicates that the energy cost is also added as a performance metric.

\subsection{Interval-Indexed Formulation}
We now modify and extend the interval-indexed formulation proposed by Hall et al. \cite{Hall1997a} to accommodate speeds and energy cost.

The interval-indexed formulation divides the time horizon into geometrically increasing intervals, and the completion time of each job is assigned to one of these intervals. Since the completion times are not associated to a specific time, the completion times are not precisely known but are lower bounded. By controlling the growth of each interval one can obtain a sufficiently tight bound.

The problem formulation is as follows. We divide the time horizon into the following geometrically increasing intervals: $[\kappa, \kappa]$, $(\kappa, (1+\epsilon)\kappa]$, $((1+\epsilon)\kappa, (1+\epsilon)^2\kappa]$, $\ldots$, where $\epsilon >0$ is an arbitrary small constant, and $\kappa = \frac{\rho_{\min}}{\mspd_{\max}}$ denotes the smallest interval size that will hold at least one whole job. We define interval $I_t = (\tau_{t-1}, \tau_t]$, with $\tau_0 = \kappa$ and $\tau_t = \kappa(1 + \epsilon)^{t-1}$. The interval index ranges over $\{1, \ldots, T\}$, with $T = \min\{\lceil t\rceil: \kappa(1+\epsilon)^{t-1} \geq \max_{i=1}^n r_i + \sum_{i=1}^n \frac{\rho_i}{\mspd_1}\}$; and thus, we have a polynomial number of indices $t$.

Let
\begin{eqnarray}
x_{ijt} = \left\{ 
\begin{array}{l l}
1, & \quad \mbox{if job $i$ runs at a speed $\mspd_j$ and completes in the time interval $I_t = (\tau_{t-1}, \tau_t]$}\\
0, & \quad \mbox{otherwise}\\ \end{array}
\right..
\label{eq:x_def}
\end{eqnarray}

By using the lower bounds $\tau_{t-1}$ of each time interval $I_t$, a lower bound to (\ref{eq:general_costfunction}) is written as,
\begin{eqnarray}
\min_{\vk{x}} \sum_{i=1}^n\sum_{j=1}^m\sum_{t=1}^T \left(v_i\rho_i\mspd_j^{\pw-1} + w_i\tau_{t-1}\right)x_{ijt}.
\label{eq:II_objective}
\end{eqnarray}

The following are the constraints required for the \ect{r_i, prec}~problem:
\begin{enumerate}
	\citem Each job must finish in a unique time interval and speed; therefore for $i = \{1,\ldots,n\}$:
	\begin{eqnarray}
	\sum_{j=1}^m\sum_{t=1}^T x_{ijt} = 1.
	\label{eq:II_ctr1}
	\end{eqnarray}\vspace{-.5cm}
	
	\citem Since only one job can be processed at any given time, the total processing time of jobs up to time interval $I_t$ must be at most $\tau_t$ units. Thus, for $t = \{1,\ldots,T\}$:
	\begin{eqnarray}
	\sum_{i=1}^n\sum_{j=1}^m\sum_{u=1}^t \frac{\rho_i}{\mspd_j} x_{iju} \leq \tau_t.
	\label{eq:II_ctr2}
	\end{eqnarray}\vspace{-.5cm}
	
	\citem Job $i$ running at speed $\mspd_j$ requires $\frac{\rho_i}{\mspd_j}$ time units to be processed, and considering that its release time is $r_i$, then for $i = \{1,\ldots,n\}$, $j = \{1,\ldots,m\}$, and $t = \{1,\ldots,T\}$:
	\begin{eqnarray}
	x_{ijt} = 0, ~~\mathrm{if} ~ \tau_t < r_i + \frac{\rho_i}{\mspd_j}.
	\label{eq:II_ctr4}
	\end{eqnarray}\vspace{-.5cm}
		
	\citem For $i = \{1,\ldots,n\}$ and $t = \{1,\ldots,T\}$:
	\begin{eqnarray}
	x_{it} \in \set{0, 1}.
	\label{eq:II_ctrbin}
	\end{eqnarray}\vspace{-.5cm}
	
	\citem The precedence constraint $i_1\prec i_2$ implies that job $i_2$ cannot finish in an interval earlier than $i_1$. Therefore for every $i_1\prec i_2$ constraint we have that for $t = \{1,\ldots,T\}$:
	\begin{eqnarray}
	\sum_{j=1}^m\sum_{u=1}^t x_{i_1ju} \geq \sum_{j=1}^m\sum_{u=1}^t x_{i_2ju}.
	\label{eq:II_ctr5}
	\end{eqnarray}\vspace{-.3cm}
\end{enumerate}

It is important to note that this integer program only provides a lower bound for (\ref{eq:general_costfunction}); in fact its optimal solution may not be schedulable, since constraints (\ref{eq:II_ctr2}) do not imply that only one job can be processed at a single time, they only bound the total amount of work in $\cup_t I_t$.

\section{Approximation Algorithm for Weighted Completion Time}
\label{sec:approx_alg}

We now describe the approximation algorithm for the weighted completion time, called \proc{\algIlong} (\algI) which is displayed in Figure \ref{alg:saias}.
\begin{figure}[b!]
	\centering
	\algbox{
		\begin{codebox}	
			\ProcStart{\algIlong~(\algI)}
			\zi {\bf Inputs: } set of jobs, $\alpha \in (0,1)$, $\epsilon > 0$, set of speeds $\vk{S} = \{\mspd_1, \ldots, \mspd_m\}$.
			\li Compute an optimal solution $\vk{\xb}$ to the linear relaxation (\ref{eq:II_objective})-(\ref{eq:II_ctr5}).
			\li Compute the $\alpha$-intervals $\vk{\tau^{\alpha}}$ and the sets $J_t$.
			\li Compute an order $\p{}^{\alpha}$ that has the sets $J_t$ ordered in non-decreasing values of $t$ and the\\ jobs within each set in a manner consistent with the precedence constraints.
			\li Compute the $\alpha$-speeds $\vk{\spd^\alpha}$
			\li Round down each $\spd_i^{\alpha}$ to the nearest speed in $\vk{S}$ and run job $i$ at this rounded speed, $\bar{\spd}_i^{\alpha}$.
			\li Set the $i$-th job to start at time $\max\{r_{\p{i}}, \bar{C}_{\p{i-1}}^\alpha\}$, where $\bar{C}_{\p{i-1}}^\alpha$ is the completion\\ time of the previous job using the rounded $\alpha$-speeds, and $\bar{C}_{\p{0}}^\alpha = 0$.
			\li \Return speeds $\bar{\vk{\spd}}^{\alpha}$ and schedule $(\p{}^{\alpha}, \bar{\vk{C}}^{\alpha})$.
		\end{codebox}
	}
	\caption{\algIlong}
	\label{alg:saias}
\end{figure}

Let $\xb_{ijt}$ denote the optimal solution of the linear relaxation of the integer program (\ref{eq:II_objective})-(\ref{eq:II_ctr5}), in which we change constraints (\ref{eq:II_ctrbin}) for $x_{ijt}\geq 0$. In step 1 of the algorithm we compute the optimal solution $\vk{\xb}$ and in step 2, given $0\leq\alpha\leq 1$, we compute the $\alpha$-interval of job $i$, which is defined as,
\begin{eqnarray}
\tau_i^\alpha = \min\left\{\tau: \sum_{j=1}^m\sum_{u=1}^\tau \xb_{iju} \geq \alpha\right\}.
\label{eq:alphaint_def}
\end{eqnarray}

Since several jobs may finish in the same interval, let $J_t$ denote the set of jobs that finish in interval $I_t$, $J_t = \{i: \tau_i^\alpha = t\}$, and we use these sets to determine the order $\p{}^{\alpha}$ as described in step 3.

Next, in step 4, we compute the $\alpha$-speeds as follows. Since $\sum_{j=1}^m\sum_{u=1}^{\tau_i^{\alpha}} \xb_{iju} \geq \alpha$, we define auxiliary variable $\{\xt_{ijt}\}$ as:
\begin{eqnarray}
\tilde{x}_{ijt} = \left\{ 
\begin{array}{l l}
\xb_{ijt}, & t < \tau_i^{\alpha}\\
\max\left\{\min\left\{\xb_{ij\tau_i^{\alpha}}, \alpha - \sum_{l=1}^{j-1}\xb_{il\tau_i^{\alpha}} - \beta_i\right\},0\right\}, & t = \tau_i^{\alpha}\\
0, & t > \tau_i^{\alpha}\\\end{array}
\right.,
\label{eq:xtilde_def}
\end{eqnarray}
where $\beta_i = \sum_{j=1}^m\sum_{u=1}^{\tau_i^{\alpha}-1} \xb_{iju} < \alpha$. Note that with this auxiliary variable $\sum_{j=1}^m\sum_{u=1}^{\tau_i^{\alpha}} \xt_{iju} = \alpha$. This is a key step that allows us to truncate the fractional solution so that for every job $i$, the sum of $\tilde{x}_{ijt}$ up to time interval $\tau_i^{\alpha}$ for each speed $j$ can be interpreted as a probability mass function. We define this probability mass function (pmf) $\mu_i = (\mu_{i1}, \dots, \mu_{im})$ on the set of speeds $\vk{S} = \{\mspd_1, \ldots, \mspd_m\}$ as
\begin{eqnarray}
\mu_{ij} = \frac{1}{\alpha}\sum_{u=1}^{\tau_i^{\alpha}}\tilde{x}_{iju}.
\label{eq:II_mu_def}
\end{eqnarray}

Let $\hat{\spd}_i$ define a random variable distributed according to the pmf $\mu_i$, i.e. $\mu_{ij} = \mathbb{P}(\hat{\spd}_i = \mspd_j)$. Then, the $\alpha$-speed of job $i$, $\spd_i^{\alpha}$, is defined as follows:
\begin{eqnarray}
\frac{1}{\spd_i^{\alpha}} = \mathbb{E} \left[\frac{1}{\hat{\spd}_i}\right] = \sum_{j=1}^m \frac{\mu_{ij}}{\mspd_j} \Rightarrow \spd_i^{\alpha} = \frac{1}{\mathbb{E} \left[\frac{1}{\hat{\spd}_i}\right]}.
\label{eq:s_ialpha}
\end{eqnarray}

We define the $\alpha$-speeds using the reciprocal of the speeds since the completion times are proportional to the reciprocals instead of the speeds, and we need to bound completion times in the analysis of the algorithm.

Next, in step 5, because the $\alpha$-speeds $\spd_i^{\alpha}$ do not necessarily belong to the set of possible speeds $\vk{S}$ we round them down to $\bar{\spd}_i^{\alpha}$, which is the nearest speed in the set such that $\bar{\spd}_i^{\alpha} \leq \spd_i^{\alpha}$. The following lemma bounds the error introduced by this rounding.

\begin{lemma}
	The cost of the solution with the rounded down speeds $\bar{\vk{\spd}}^{\alpha}$ is at most $(1+\delta)$ times the cost of the solution using the $\alpha$-speeds $\vk{\spd}^{\alpha}$.
	\label{lm:speeddown_bound}
\end{lemma}
\begin{proof}
	The energy cost function $E_i(\spd_i)$ is increasing so rounding down does not increase the energy cost, but the completion time is now larger. Let $C_i^{\alpha}$ be the completion time of job $i$ when the speeds $\vk{\spd}^{\alpha}$ are used and $\bar{C}_i^{\alpha}$ when the rounded ones $\bar{\vk{\spd}}^{\alpha}$ are used. Since the speeds are reduced at most by $(1+\delta)$, then $(1+\delta)\bar{\spd_i}^{\alpha} \geq \spd_i^{\alpha}$, and we have that,
	\begin{eqnarray}
	\bar{C}_i^{\alpha} = \max\{r_i, \bar{C}_{i-1}^{\alpha}\} + \frac{\rho}{\bar{\spd}_i^{\alpha}} \leq (1+\delta)\left(\max\{r_i, C_{i-1}^{\alpha}\} + \frac{\rho}{\spd_i^{\alpha}}\right) = (1+\delta)C_i^{\alpha},
	\end{eqnarray}
	which implies that $\sum_{i=1}^nw_i\bar{C}_i^{\alpha} \leq (1+\delta)\sum_{i=1}^nw_i C_i^{\alpha}$ and proves the lemma.
\end{proof}

Finally, in steps 6 and 7 we compute the completion times given the calculated speeds and return the set of speeds $\bar{\vk{\spd}}^{\alpha}$ and the schedule $(\p{}^{\alpha}, \bar{\vk{C}}^{\alpha})$.

We now analyse this algorithm's performance for different energy aware scheduling problems. In the following subsections we will assume w.l.o.g. that $\tau_1^\alpha \leq \tau_2^\alpha \leq \ldots \tau_n^\alpha$.

\subsection{Single Machine Problem with Precedence Constraints}

We first need to prove that the output of the \algI~algorithm is indeed feasible.
\begin{lemma}
	If $i_1\prec i_2$, then constraint (\ref{eq:II_ctr5}) implies that $\tau_{i_1}^\alpha \leq \tau_{i_2}^\alpha$.
	\label{lm:II_prec_ctr}
\end{lemma}
\begin{proof}
	Evaluating the LP constraint (\ref{eq:II_ctr5}) corresponding to $i_1\prec i_2$, for $t=\tau_{i_2}^\alpha$, we have that,
	\begin{eqnarray*}
	\sum_{j=1}^m\sum_{u=1}^{\tau_{i_2}^\alpha} x_{i_1ju} \geq \sum_{j=1}^m\sum_{u=1}^{\tau_{i_2}^\alpha} x_{i_2ju} \geq \alpha,
	\end{eqnarray*}
	where the last inequality follows from the definition of $\tau_{i_2}^\alpha$. The chain of inequalities implies that $\sum_{j=1}^m\sum_{u=1}^{\tau_{i_2}^\alpha} x_{i_1 j u} \geq \alpha$, so $\tau_{i_1}^\alpha \leq \tau_{i_2}^\alpha$.
\end{proof}

Since the \algI~algorithm schedules jobs by first ordering the sets $J_t$ in increasing order of $t$, and then orders the jobs within each set in a way that is consistent with the precedence constraints, by Lemma \ref{lm:II_prec_ctr} it follows that the \algI~algorithm preserves the precedence constraints, and, therefore, the output of the algorithm is feasible. Next, we can prove the following result.

\begin{theorem}
	The \algI~algorithm with $\alpha=\frac{1}{2}$ is a $4(1+\epsilon)(1 + \delta)$-approximation algorithm for the \ect{prec}~problem, with $E_i(\spd_i)=v_i\rho_i\spd_i^{\pw -1}$.
	\label{th:II_nc}
\end{theorem}
\begin{proof}
	Let $x_{ijt}^*$ denote an optimal solution to the integer problem (\ref{eq:II_objective})-(\ref{eq:II_ctr5}), $\xb_{ijt}$ the fractional solution of its linear relaxation, and $\xt_{iju}$ the auxiliary variables calculated for the \algI~algorithm.
	
	Since in (\ref{eq:II_objective}) the completion time for jobs completed in interval $I_t$ is $\tau_{t-1}$, it follows that,
	\begin{eqnarray}
	\sum_{i=1}^n\sum_{j=1}^m\sum_{t=1}^T \left(v_i \rho_i \mspd_j^{\pw-1} + w_i \tau_{t-1}\right) \xb_{ijt}
	\leq \sum_{i=1}^n\sum_{j=1}^m\sum_{t=1}^T v_i \rho_i \mspd_j^{\pw-1} x_{ijt}^* + \sum_{i=1}^n w_i C_i^*.
	\label{eq:II_xopt_bound}
	\end{eqnarray}
	
	The energy terms of the algorithm's solution are bounded as follows,
	\begin{eqnarray}
	v_i \rho_i(\spd_i^\alpha)^{\pw-1} & = & v_i \rho_i \left(\frac{1}{\spd_i^\alpha}\right)^{-(\pw-1)} = v_i \rho_i \left(\mathbb{E}\left[\frac{1}{\hat{\spd}_i}\right]\right)^{-(\pw-1)}\nonumber\\
	& \leq & v_i \rho_i \mathbb{E}\left[\left(\frac{1}{\hat{\spd}_i}\right)^{-(\pw-1)}\right] = v_i \rho_i \mathbb{E}\left[\hat{\spd}_i^{\pw-1}\right] = v_i \rho_i \sum_{j=1}^m \mu_{ij} \mspd_j^{\pw-1},
	\label{eq:energy_terms}
	\end{eqnarray}
	where the inequality follows from Jensen's Inequality applied to the convex function $\frac{1}{s^{\pw-1}}$. Using the definition of $\mu_{ij}$ in (\ref{eq:II_mu_def}) and given that $0\leq \alpha \leq 1$, $\epsilon > 0$, and $\xt_{ijt}\leq\xb_{ijt}$, it follows that,
	\begin{eqnarray}
	v_i \rho_i (\spd_i^\alpha)^{\pw-1} \leq \frac{v_i\rho_i}{\alpha} \sum_{j=1}^m \sum_{u=1}^{\tau_i^\alpha} \mspd_j^{\pw-1}\xt_{iju} \leq \frac{(1+\epsilon)}{\alpha(1-\alpha)} v_i\rho_i \sum_{j=1}^m \sum_{u=1}^{T} \mspd_j^{\pw-1}\xb_{iju}.
	\label{eq:II_e_term}
	\end{eqnarray}
	
	Since there are no release date constraints there is no idle time between jobs,
	\begin{eqnarray}
	C_i^\alpha = \sum_{j=1}^i \frac{\rho_j}{\spd_j^\alpha} = \sum_{j=1}^i \rho_j \mathbb{E}\left[\frac{1}{\hat{\spd_j}}\right] = \frac{1}{\alpha} \sum_{j=1}^i \sum_{l=1}^m \sum_{u=1}^{\tau_j^\alpha} \frac{\rho_j}{\mspd_l} \xt_{jlu} \leq \frac{1}{\alpha} \sum_{j=1}^n \sum_{l=1}^m \sum_{u=1}^{\tau_i^\alpha} \frac{\rho_j}{\mspd_l} \xb_{jlu},
	\label{eq:II_c_alpha_def}
	\end{eqnarray}
	and from constraint (\ref{eq:II_ctr2}) for $t = \tau_i^\alpha$ we get,
	$C_i^\alpha \leq \frac{1}{\alpha} \tau_{\tau_i^\alpha}$.
	
	Let $\bar{C}_i = \sum_{j=1}^m\sum_{t=1}^T \tau_{t-1} \xb_{ijt}$ denote the optimal fractional completion time given by the optimal solution of the relaxed linear program (\ref{eq:II_objective})-(\ref{eq:II_ctr4}). Since it is possible that $\sum_{j=1}^m\sum_{t=1}^{\tau_i^\alpha} \xb_{ijt} > \alpha$; we define $X_{i}^{(1)} = \alpha - \sum_{j=1}^m\sum_{t=1}^{\tau_i^\alpha - 1} \xb_{ijt}$ and $X_{i}^{(2)} = \sum_{j=1}^m\sum_{t=1}^{\tau_i^\alpha} \xb_{ijt} - \alpha$, thus $X_{i}^{(1)} + X_{i}^{(2)} = \sum_{j=1}^m \xb_{ij {\tau_i^\alpha}}$, and we can rewrite
	\begin{eqnarray}
	\bar{C}_i = \sum_{j=1}^m\sum_{t=1}^{\tau_i^\alpha-1} \tau_{t-1} \xb_{ijt} + \tau_{\tau_i^\alpha-1}X_{i}^{(1)} + \tau_{\tau_i^\alpha-1}X_{i}^{(2)} + \sum_{j=1}^m\sum_{t=\tau_i^\alpha+1}^{T} \tau_{t-1} \xb_{ijt},
	\end{eqnarray}
	and eliminating the lower terms of the previous sum we get that,
	\begin{eqnarray}
	\bar{C}_i \geq \tau_{\tau_i^\alpha-1}X_{i}^{(2)} + \sum_{j=1}^m\sum_{t=\tau_i^\alpha+1}^{T} \tau_{t-1} \xb_{ijt} \geq \tau_{\tau_i^\alpha-1}X_{i}^{(2)} + \sum_{j=1}^m\sum_{t=\tau_i^\alpha+1}^{T} \tau_{\tau_i^\alpha-1} \xb_{ijt} = \tau_{\tau_i^\alpha-1} (1-\alpha).
	\label{eq:II_cbar_bound}
	\end{eqnarray}
	
	Because $\tau_{\tau_i^\alpha} = (1+\epsilon)\tau_{\tau_i^\alpha - 1}$, from (\ref{eq:II_c_alpha_def}) and (\ref{eq:II_cbar_bound}) we get that $C_i^\alpha \leq \frac{(1+\epsilon)}{\alpha(1-\alpha)} \bar{C}_i \Rightarrow \sum_{i=1}^n w_iC_i^\alpha \leq \frac{(1+\epsilon)}{\alpha(1-\alpha)} \sum_{i=1}^n w_i \bar{C}_i$. From this, (\ref{eq:II_xopt_bound}) and (\ref{eq:II_e_term}) it follows that,
	\begin{eqnarray}
	\sum_{i=1}^n v_i \rho_i (\spd_i^\alpha)^{\pw-1} + \sum_{i=1}^n w_i C_i^{\alpha} \leq \frac{(1+\epsilon)}{\alpha(1-\alpha)} \left[\sum_{i=1}^n\sum_{j=1}^m\sum_{t=1}^T v_i \rho_i \mspd_j^{\pw-1} x_{ijt}^* + \sum_{i=1}^n w_i C_i^*\right],
	\end{eqnarray}
	and we set $\alpha = \arg\min_{0\leq\alpha\leq 1}\left\{\frac{1}{\alpha(1-\alpha)} \right\} = \frac{1}{2}$, to minimize the bound. By Lemma \ref{lm:speeddown_bound}, which bounds the final rounding error, we get the desired approximation ratio.
\end{proof}

\subsection{Single Machine Problem with Precedence and Release Date Constraints}
We now analyse the case with precedence constraints and release dates. Release dates makes the problem somewhat harder since they can introduce idle times between jobs.

\begin{theorem}
	The \algI~algorithm with $\alpha=\sqrt{2} - 1$ is a $(3 + 2\sqrt{2})(1+\epsilon)(1 + \delta)$-approximation algorithm for the \ect{r_i, prec}~problem, with $E_i(\spd_i)=v_i\rho_i\spd_i^{\pw -1}$.
	\label{th:II_precr}
\end{theorem}
\begin{proof}
	The bound for the energy terms computed in equation (\ref{eq:energy_terms}) are still valid when there is idle time between jobs, we have that,
	\begin{eqnarray}
	v_i \rho_i (\spd_i^\alpha)^{\pw-1} \leq \frac{(1+\epsilon)}{\alpha(1-\alpha)} v_i\rho_i \sum_{j=1}^m \sum_{u=1}^{T} \mspd_j^{\pw-1}\xb_{iju} \leq \frac{(1+\epsilon)(1+\alpha)}{\alpha(1-\alpha)} v_i\rho_i \sum_{j=1}^m \sum_{u=1}^{T} \mspd_j^{\pw-1}\xb_{iju}.
	\label{eq:II_e_term2}
	\end{eqnarray}
	
	When bounding the completion time $C_i^\alpha$, given the sorting done in step 3 of the \algI~algorithm, now one has to consider all the jobs up to the ones in set $J_{\tau_i^\alpha}$, and thus,
	\begin{eqnarray}
	C_i^\alpha \leq \max_{j\in\{J_1,\ldots,J_{\tau_i^\alpha}\}} r_j + \sum_{j\in\{J_1,\ldots,J_{\tau_i^\alpha}\}}\frac{\rho_j}{\spd_j^\alpha}.
	\label{eq:II_calpha_bound1}
	\end{eqnarray}
	
	Since all jobs that have been at least partially processed up to time interval $I_t$ need to be released before $\tau_t$, it follows that $\max_{j\in\{J_1,\ldots,J_{\tau_i^\alpha}\}} r_j \leq \tau_{\tau_i^\alpha}$. On the other hand, we also have that,
	\begin{eqnarray}
	\sum_{j\in\{J_1,\ldots,J_{\tau_i^\alpha}\}} \frac{\rho_j}{\spd_j^\alpha} = \frac{1}{\alpha} \sum_{j\in\{J_1,\ldots,J_{\tau_i^\alpha}\}} \sum_{l=1}^m \sum_{u=1}^{\tau_j^\alpha} \frac{\rho_j}{\mspd_l} \xt_{jlu} \leq \frac{1}{\alpha} \sum_{j=1}^n \sum_{l=1}^m \sum_{u=1}^{\tau_i^\alpha} \frac{\rho_j}{\mspd_l} \xb_{jlu} \leq \frac{1}{\alpha} \tau_{\tau_i^\alpha},
	\end{eqnarray}
	where the last inequality follows from constraint (\ref{eq:II_ctr2}) with $t=\tau_i^\alpha$. Thus, $C_i^\alpha \leq \frac{(1+\alpha)}{\alpha}\tau_{\tau_i^\alpha}$. Since $\bar{C}_i = \sum_{j=1}^m\sum_{t=1}^T \tau_{t-1} \xb_{ijt}$, (\ref{eq:II_cbar_bound}) is still valid and because $\tau_{\tau_i^\alpha} = (1+\epsilon) \tau_{\tau_i^\alpha - 1}$, we get,
	\begin{eqnarray}
	C_i^\alpha \leq \frac{(1+\epsilon)(1+\alpha)}{\alpha(1-\alpha)}\bar{C}_i \Rightarrow \sum_{i=1}^n w_iC_i^\alpha \leq \frac{(1+\epsilon)(1+\alpha)}{\alpha(1-\alpha)} \sum_{i=1}^n w_i\bar{C}_i.
	\label{eq:II_calpha_bound3}
	\end{eqnarray}
	
	Finally, from (\ref{eq:II_e_term2}) and (\ref{eq:II_calpha_bound3}) it follows that,
	\begin{eqnarray}
	\sum_{i=1}^n v_i \rho_i (\spd_i^\alpha)^{\pw-1} + \sum_{i=1}^n w_i C_i^{\alpha}
	\leq \frac{(1+\epsilon)(1+\alpha)}{\alpha(1-\alpha)} \left[\sum_{i=1}^n\sum_{j=1}^m\sum_{t=1}^T v_i \rho_i \mspd_j^{\pw-1} x_{ijt}^* + \sum_{i=1}^n w_i C_i^*\right],
	\end{eqnarray}
	and by setting $\alpha = \arg\min_{0\leq\alpha\leq 1} \left\{\frac{(1+\alpha)}{\alpha(1-\alpha)} \right\} = \sqrt{2} - 1$, and again using Lemma \ref{lm:speeddown_bound} to bound the speed-rounding error, we get the required approximation ratio.
\end{proof}

If no precedence constraints and release dates exist, there are two versions of this problem that can be optimally solved in polynomial time: when all weights $w_i$ are equal, and when all jobs are of the same size (i.e. $\rho_i = \rho$, $\forall i$) and all jobs have the same energy cost function. For these cases we have the following result:
\begin{theorem}
If $w_i = w, ~\forall i$ or $\rho_iv_i^{\frac{1}{\pw}} = \xi, ~\forall i$ then the order $\p{}$ is optimal if
$$\frac{w_{\pi(i)}}{\rho_{\pi(i)}v_{\pi(i)}^{\frac{1}{\pw}}} \geq \frac{w_{\pi(i+1)}}{\rho_{\pi(i+1)}v_{\pi(i+1)}^{\frac{1}{\pw}}}, ~\forall i\in\set{1,\ldots,n-1}.$$
\label{th:order_nc}
\end{theorem}
\begin{proof}
	For simplicity we will define $\xi_i \equiv \rho_iv_i^{\frac{1}{k}}$, $q = \frac{k-1}{k}$, and $\mathpzc{K} \equiv \frac{k}{(k-1)^{\frac{k-1}{k}}}$. First, dual formulation of problem (\ref{eq:general_costfunction}) with no precedence or release date constraints is given by,
	\begin{eqnarray}
	\min_{\pi} F(\pi) = \min_{\pi} \sum_{i=1}^n \mathpzc{K} \xi_{\pi(i)}\left(\sum_{j=i}^nw_{\pi(j)} \right)^{q}.
	\label{eq:dual}
	\end{eqnarray}
	
	We now prove both cases by contradiction using the dual formulation.
	
	When $w_i = w, ~\forall i$, Theorem \ref{th:order_nc} implies that in the optimal order $\xi_{\pi(i+1)}\geq\xi_{\pi(i)}$. By contradiction, let $\pi$ be an optimal order such that for some index $k$, $\xi_{\pi(k+1)} < \xi_{\pi(k)}$. For this order the total cost is
	\begin{eqnarray*}
	F(\pi) & = & \sum_{i=1}^n \mathpzc{K} \xi_{\pi(i)}\left(\sum_{j=i}^nw \right)^{q} = \sum_{i=1}^n \mathpzc{K} \xi_{\pi(i)}\left((n-i+1)w \right)^{q},\\
	& = & \mathpzc{K}w^q \left\{\xi_{\pi(k)}(n-k+1)^q + \xi_{\pi(k+1)}(n-k)^q + \sum_{i=1;~i\neq k, k+1}^n \xi_{\pi(i)}(n-i+1)^{q}\right\}.
	\end{eqnarray*}

	Let $\pi_k$ define the order where we switch jobs $k$ and $k+1$ from order $\pi$, i.e. $\pi_k(k) = \pi(k+1)$ and $\pi_k(k+1) = \pi(k)$. Given this order we have that
	\begin{eqnarray*}
	F(\pi) - F(\pi_k) & = & \mathpzc{K}w^q \left\{\xi_{\pi(k)}(n-k+1)^q + \xi_{\pi(k+1)}(n-k)^q - \xi_{\pi(k+1)}(n-k+1)^q - \xi_{\pi(k)}(n-k)^q\right\},\\
	& = & \mathpzc{K}w^q \left\{(n-k+1)^q(\xi_{\pi(k)} - \xi_{\pi(k+1)}) - (n-k)^q(\xi_{\pi(k)} - \xi_{\pi(k+1)})\right\},\\
	& = & \mathpzc{K}w^q \left\{\left(\xi_{\pi(k)} - \xi_{\pi(k+1)}\right)\left((n-k+1)^q - (n-k)^q\right)\right\}.
	\end{eqnarray*}

	By our initial assumption the first term is positive (since $\xi_{\pi(k+1)} < \xi_{\pi(k)}$) and the second one is always positive, hence $F(\pi) - F(\pi_k) > 0$ which is a contradiction, since that implies that $\pi_k$ has a smaller cost.

	For the case when $\xi_i=\xi, ~\forall i$, Theorem \ref{th:order_nc} implies that an order $\pi$ is optimal then $w_{\pi(i)}\geq w_{\pi(i+1)}$. Let $\pi$ be an optimal order  such that for some index $k$, $w_{\pi(k)} < w_{\pi(k+1)}$. The total cost for this solution is
	\begin{eqnarray*}
	F(\pi) & = & \sum_{i=1}^n \mathpzc{K} \xi\left(\sum_{j=i}^nw_{\pi(i)} \right)^{q} =  \mathpzc{K}\xi \left\{ \sum_{i=1}^k \left(\sum_{j=i}^nw_{\pi(i)} \right)^{q} + \left(\sum_{j=k+1}^nw_{\pi(i)} \right)^{q} + \sum_{i=k+2}^n \left(\sum_{j=i}^nw_{\pi(i)} \right)^{q}\right\},\\
	& = & \mathpzc{K}\xi \left\{ \sum_{i=1}^k \left(\sum_{j=i}^nw_{\pi(i)} \right)^{q} + \left(w_{\pi(k+1)} + \sum_{j=k+2}^nw_{\pi(i)} \right)^{q} + \sum_{i=k+2}^n \left(\sum_{j=i}^nw_{\pi(i)} \right)^{q}\right\}.
	\end{eqnarray*}

Let $\pi_k$ define the order where we switch jobs $k$ and $k+1$ from order $\pi$. Given this new order we have
\begin{eqnarray*}
F(\pi) - F(\pi_k) & = & \mathpzc{K}\xi \left\{ \left(w_{\pi(k+1)} + \sum_{j=k+2}^nw_{\pi(i)} \right)^{q} - \left(w_{\pi(k)} + \sum_{j=k+2}^nw_{\pi(i)} \right)^{q}\right\} > 0,
\end{eqnarray*}
since $w_{\pi(k+1)} > w_{\pi(k)}$ by our initial assumption, which is a contradiction since this result implies that order $\pi_k$ has a lower cost.
\end{proof}

\section{Extension to the Weighted Tardiness Problem}
\label{sec:tardiness}

In this section we extend our results to the weighted tardiness setting. We still allow for arbitrary precedence constraints but no release dates. In this case, each job $i$ also has a deadline $d_i$. The tardiness $T_i$ of job $i$ is defined as $T_i = \max\{0, C_i-d_i\}$, and the objective function is now given by,
\begin{eqnarray}
g(\p{}, \vk{\spd}) = \sum_{i=1}^n v_i\rho_i\spd_i^{\pw-1} + \sum_{i=1}^n w_{\p{i}}\left(C_{\p{i}} - d_{\p{i}}\right)^{+}.
\label{eq:due_date_costfunction}
\end{eqnarray}

We now formulate the problem using a modification of the interval-and-speed-indexed formulation presented in Section \ref{sec:formulation}. Because the completion time can be bounded by $\sum_{j=1}^m\sum_{t=1}^T \tau_{t-1} x_{ijt}$, we can bound (\ref{eq:due_date_costfunction}) from below by the following optimization problem,
\begin{eqnarray}
\min_{\vk{x}} \sum_{i=1}^n\sum_{j=1}^m\sum_{t=1}^T \left(v_i\rho_i\mspd_j^{\pw-1} + w_i\left(\tau_{t-1} - d_i\right)^+ \right)x_{ijt},
\label{eq:IIdue_objective}
\end{eqnarray}
together with constraints (\ref{eq:II_ctr1})-(\ref{eq:II_ctr5}) from the interval-indexed formulation. Note that although the objective (\ref{eq:due_date_costfunction}) is non-linear, because we have a interval-indexed formulation, (\ref{eq:IIdue_objective}) is linear.

\begin{figure}[t!]
	\centering
	\algbox{
		\begin{codebox}
			\ProcStart{\algduelong~(\algdue)}
			\zi {\bf Inputs: } set of jobs, $\alpha \in (0,1)$, $\epsilon > 0$, $\gamma > 1$, set of speeds $\vk{S} = \{\mspd_1, \ldots, \mspd_m\}$.
			\li Compute an optimal solution $\vk{\xb}$ to the linear relaxation (\ref{eq:IIdue_objective}), (\ref{eq:II_ctr1})-(\ref{eq:II_ctr5}).
			\li Compute the $\alpha$-intervals $\vk{\tau^{\alpha}}$ and the sets $J_t$ as in the \algI~algorithm.
			\li Compute an order $\p{}^{\alpha}$ that has the sets $J_t$ ordered in non-decreasing values of $t$ and the\\ jobs within each set in a manner consistent with the precedence constraints.
			\li Compute the $\alpha$-speeds $\vk{\spd^\alpha}$ and scale each $\spd_i^{\alpha}$ to $\tilde{\spd}_i^\alpha = \gamma\spd_i^{\alpha}$.
			\li Round up each $\tilde{\spd}_i^\alpha$ to the next speed in $\vk{S}$, $\bar{\spd}_i^\alpha$ and run each job $i$ at this new speed.
			\li Set the $i$-th job to start at time $\max\{r_{\p{i}}, \bar{C}_{\p{i-1}}^\alpha\}$, where $\bar{C}_{\p{i-1}}^\alpha$ is the completion\\ time of the previous job using the rounded $\alpha$-speeds, and $\bar{C}_{\p{0}}^\alpha = 0$.
			\li \Return speeds $\bar{\vk{\spd}}^{\alpha}$ and schedule $(\p{}^{\alpha}, \bar{\vk{C}}^{\alpha})$.
		\end{codebox}
	}
	\caption{\algduelong~Algorithm}
	\label{alg:saias-t}
\end{figure}
We approximately solve (\ref{eq:due_date_costfunction}) using the \proc{\algduelong} (\algdue) Algorithm displayed in Figure \ref{alg:saias-t}. The main difference with the \algI~algorithm, is that in step 4 we scale up the $\alpha$-speeds. This scaling makes the completion time of the relaxed LP comparable to the completion time of the algorithm's output, and thus jobs that have $0$ tardiness in the LP also have $0$ tardiness in our algorithm. If we rounded speeds down, jobs with $0$ tardiness in the LP could, at a lower speed, miss their deadline, and thus the approximation ratio could be arbitrary large.

We now analyse the algorithm assuming w.l.o.g. that $\tau_1^{\alpha}\leq \tau_2^{\alpha}\leq \ldots \leq \tau_n^{\alpha}$. Since Lemma \ref{lm:II_prec_ctr} remains valid, arguments identical to those in Section \ref{sec:formulation} show that the output of the \algdue~algorithm is feasible; thus, we have the following theorem:
\begin{theorem}
	The \algdue~algorithm with $\gamma = \frac{(1+\epsilon)}{\alpha(1-\alpha)}$ and $\alpha=\frac{1}{2}$ is a $4^{\pw}(1+\epsilon)^{\pw-1}(1 + \delta)^{\pw-1}$-approximation algorithm for the \etr{prec}~problem, with $E_i(\spd_i)=v_i\rho_i\spd_i^{\pw -1}$.
	\label{th:IIdue_nc}
\end{theorem}
\begin{proof}
	Let $\bar{C}_i = \sum_{j=1}^m \sum_{t=1}^T \tau_{t-1}\xb_{ijt}$ denote the optimal fractional completion time of the relaxed linear program. $(\bar{C}_i - d_i)^{+}$ is a lower bound for the optimal tardiness $(C^*_i - d_i)^{+}$, since $\sum_{jt}(\tau_{t-1}-d_i)^{+}\xb_{ijt} \geq (\bar{C}_i - d_i)^{+}$. Thus,
	\begin{eqnarray}
	\sum_{i=1}^n\sum_{j=1}^m\sum_{t=1}^T v_i\rho_i\mspd_j^{\pw-1}\xb_{ijt} + \sum_{i=1}^n w_i \left(\bar{C}_i - d_i\right)^{+} \leq \sum_{i=1}^n\sum_{j=1}^m\sum_{t=1}^T v_i\rho_i\mspd_j^{\pw-1} x_{ijt}^{*} + \sum_{i=1}^n w_i \left(C_i^{*} - d_i\right)^{+}.
	\label{eq:IIdue_xopt_bound}
	\end{eqnarray}
	
	Let $\tilde{C}_i^{\alpha}$ denote the completion time of job $i$ using speeds $\tilde{\vk{\spd}}^{\alpha}$ and $C_i^{\alpha}$ the one using speeds $\vk{\spd}^{\alpha}$. Because there are no release date constraints, there is no idle time in between jobs; therefore,
	\begin{eqnarray}
	\tilde{C}_i^{\alpha} = \sum_{j=1}^i \frac{\rho_j}{\tilde{\spd}_j^{\alpha}} = \frac{1}{\gamma} \sum_{j=1}^i \frac{\rho_j}{\spd_j^\alpha} = \frac{1}{\gamma} C_i ^{\alpha}.
	\label{eq:IIdue_c_alpha}
	\end{eqnarray}
	
	Since (\ref{eq:II_c_alpha_def}) remains valid, it follows that $C_i^{\alpha} \leq \frac{(1+\epsilon)}{\alpha(1-\alpha)} \bar{C}_i ~\Rightarrow~ \tilde{C}_i^{\alpha} \leq \frac{1}{\gamma} \frac{(1+\epsilon)}{\alpha(1-\alpha)} \bar{C}_i$. The key step is that by setting $\gamma = \frac{(1+\epsilon)}{\alpha(1-\alpha)}$, which makes the two completion times comparable, we have that,
	\begin{eqnarray}
	\sum_{i=1}^n w_i\left(\tilde{C}_i^{\alpha} -d_i\right)^{+} \leq \sum_{i=1}^n w_i \left(\frac{1}{\gamma} \frac{(1+\epsilon)}{\alpha(1-\alpha)}\bar{C}_i - d_i\right)^{+} = \sum_{i=1}^n w_i \left(\bar{C}_i - d_i\right)^{+}.
	\label{eq:IIdue_ctilde}
	\end{eqnarray}
	
	The energy term is bounded in a manner analogous to (\ref{eq:II_e_term}):
	\begin{eqnarray}
	v_i\rho_i(\tilde{\spd}_i^\alpha)^{\pw-1} = \gamma^{\pw-1} v_i\rho_i(\spd_i^\alpha)^{\pw-1} \leq \frac{(1+\epsilon)^{\pw-1}}{\left(\alpha(1-\alpha)\right)^{\pw}} v_i\rho_i \sum_{j=1}^m\sum_{t=1}^T \mspd_j^{\pw-1}\xb_{ijt},
	\label{eq:IIdue_e_term}
	\end{eqnarray}
	where the last inequality follows from (\ref{eq:II_e_term}) that remains valid.
	
	From (\ref{eq:IIdue_xopt_bound}), (\ref{eq:IIdue_e_term}), and (\ref{eq:IIdue_ctilde}) it follows that,
	\begin{eqnarray}
	\sum_{i=1}^n v_i \rho_i (\tilde{\spd}_i^\alpha)^{\pw-1} + \sum_{i=1}^n w_i \left(\tilde{C}_i^{\alpha} - d_i\right)^{+} \leq \frac{(1+\epsilon)^{\pw-1}}{\left(\alpha(1-\alpha)\right)^{\pw}} \left[\sum_{i=1}^n\sum_{j=1}^m\sum_{t=1}^T v_i \rho_i \mspd_j^{\pw-1} x_{ijt}^* + \sum_{i=1}^n w_i \left(C_i^* - d_i\right)^{+}\right]\nonumber.
	\end{eqnarray}
	
	Because speeds are rounded up, the completion times, and thus the tardiness can only improve, whereas the energy cost increases. Since at most we speed up each job by a factor $(1 + \delta)$, we have that,
	\begin{eqnarray}
	E_i(\bar{\spd}_i^{\alpha}) \leq E_i\left((1+\delta)\spd_i^{\alpha}\right) = (1+\delta)^{\pw-1}E_i(\spd_i^{\alpha}) \Rightarrow \sum_{i=1}^n E_i(\bar{\spd}_i^{\alpha}) \leq (1+\delta)^{\pw-1}\sum_{i=1}^n E_i(\spd_i^{\alpha}).
	\label{eq:IIdue_rounding}
	\end{eqnarray}
	
	The approximation ratio follows from setting $\alpha = \arg\min_{0\leq \alpha\leq 1} \left\{\frac{1}{\left(\alpha(1-\alpha)\right)^{\pw}} \right\}= \frac{1}{2}$. Clearly we could use $\frac{(1+\epsilon)^{\pw-1}}{\alpha^{\pw}(1-\alpha)^{\pw-1}}$ in (\ref{eq:IIdue_e_term}) to compute a tighter bound, but the resulting expression is not as simple.
\end{proof}

We are not able to extend this algorithm for the \etr{r_i}~problem, since it is based on speed scaling to make sure that jobs are finished within a desired time interval. When release dates are present, we do not see how  to arbitrarily reduce the completion times.

\section{Extension to General Energy Cost Functions}
\label{sec:extension}

In this section we consider the extension to general energy {\em cost} functions, as opposed to simply energy {\em consumption}. We begin by considering discrete speeds, as in the previous sections, but in Section \ref{sec:extension_cspeeds} we will relax this requirement. 

Managers of data centres are clearly interested in the energy cost metric, since they need to balance the penalty for violating the service level agreements with the cost of energy. The energy price curves for industrial consumers are often quite complicated because of energy contracts, discounts, real time pricing etc.; therefore it is very important to consider general cost functions in the scheduling model. Hence, in this section we use $\mathcal{E}_i(\spd_i)$ as the general energy cost function of running job $i$ at speed $\spd_i$. We will require that $\mathcal{E}_i(\spd_i)$ is non-negative, just as in \cite{Andrew2009, Bansal2009}, but no other requirements are needed for the weighted completion time setting. For the weighted tardiness setting we will require an additional regularity condition that bounds the growth of the energy cost function.

Since in practice the processor speed can be dynamically changed during the course of a job, one can replace the general cost function by its lower convex envelope. Hence, without loss of generality, we can assume that $\mathcal{E}_i(\spd_i)$ is convex. Furthermore, since the machine can only run at the speeds in $\vk{S}$, we can also consider that $\mathcal{E}_i(\spd)$ is linear in between these speeds. Hence, for every $\spd\in[\mspd_j, \mspd_{j+1}]$ such that $\spd = \lambda\mspd_j + (1-\lambda)\mspd_{j+1}$, with $\lambda\in[0,1]$, then  $\mathcal{E}_i(\spd) = \lambda \mathcal{E}_i(\mspd_j) + (1-\lambda)\mathcal{E}_i(\mspd_{j+1})$.

Note that for bounding the energy cost terms in the weighted completion time setting, we only used the fact that the energy consumption function $E_i(\spd) = v_i\rho_i\spd^{\pw-1}$ is convex. Thus, the previous bounds extend to our more general class of functions $\mathcal{E}_i(\spd)$. In the weighted tardiness case we required also a bound on the growth of the energy cost function, which we will address in Section \ref{sec:extension_wT}.

\subsection{Weighted Completion Time Problem with General Energy Cost}
\label{sec:extension_wC}
The objective function (\ref{eq:II_objective}) is extended as follows,
\begin{eqnarray}
\min_{\vk{x}} \sum_{i=1}^n\sum_{j=1}^m\sum_{t=1}^T \left(\mathcal{E}_i(\mspd_j) + w_i \tau_{t-1} \right)x_{ijt},
\label{eq:EX_IIobjective}
\end{eqnarray}
where $\mathcal{E}_i(\mspd_j)$ are just coefficients. Given that we only change the energy cost related terms, all the completion time related bounds computed previously are still valid.

The only modification required is in the rounding procedure at the end of the \algI~algorithm, where it was done by rounding down the $\alpha$-speeds. Now instead we will round them up or down such that $\mathcal{E}_i(\bar{\spd}_i^{\alpha}) \leq \mathcal{E}_i(\spd_i^{\alpha})$, which is always possible since $\mathcal{E}_i(\spd_i)$ is linear in between the speeds in $\vk{S}$. With this change Lemma \ref{lm:speeddown_bound} remains valid and we can extend the algorithm to our general energy cost functions.

\begin{theorem}
	The \algI~algorithm with $\alpha=\frac{1}{2}$ is a $4(1+\epsilon)(1 + \delta)$-approximation algorithm for the \gect{prec}~problem, for all general non-negative energy cost functions $\mathcal{E}_i(\spd)$.
	\label{th:EXT_II_prec}
\end{theorem}
\begin{proof}
	Because $\mathcal{E}_i(\mspd)$, $i=\{1,\ldots,n\}$ are convex functions, (\ref{eq:energy_terms}) remains valid since $\mathcal{E}_i(\spd_i^\alpha) = \mathcal{E}_i(\mathbb{E}[\hat{\spd}_i]) \leq \mathbb{E}[\mathcal{E}_i(\hat{\spd}_i)] = \sum_{j=1}^m \mu_{ij} \mathcal{E}_i(\mspd_j)$, and thus, from the definition of $\mu_{ij}$, and from $0\leq\alpha\leq 1$, $\epsilon>0$, and $\xt_{ijt}\leq\xb_{ijt}$,
	\begin{eqnarray}
	\sum_{i=1}^n \mathcal{E}_i(\spd_i^\alpha) \leq \frac{1}{\alpha} \sum_{i=1}^n \sum_{j=1}^m \sum_{t=1}^{\tau_i^\alpha}  \mathcal{E}_i(\mspd_j) \xt_{ijt} \leq \frac{(1+\epsilon)}{\alpha(1-\alpha)} \sum_{i=1}^n \sum_{j=1}^m \sum_{t=1}^T  \mathcal{E}_i(\mspd_j) \xb_{ijt}.
	\label{eq:EXT_II_e_bound}
	\end{eqnarray}
	
	The proof follows since the bounds for the completion time in Theorem \ref{th:II_nc} remain valid, as well as Lemma \ref{lm:speeddown_bound}.
\end{proof}

By the same argument we also have that,
\begin{theorem}
	The \algI~algorithm with $\alpha=\sqrt{2}-1$ is a $(3+2\sqrt{2})(1+\epsilon)(1 + \delta)$-approximation algorithm for the \gect{r_i, prec}~problem, for all general non-negative energy cost functions $\mathcal{E}_i(\spd)$.
	\label{th:EXT_II_precr}
\end{theorem}

\subsection{Weighted Tardiness Problem with General Energy Cost}
\label{sec:extension_wT}
We replace the energy term in (\ref{eq:IIdue_objective}) with the general energy cost term to obtain the new objective
\begin{eqnarray}
\min_{\vk{x}} \sum_{i=1}^n\sum_{j=1}^m\sum_{t=1}^T \left(\mathcal{E}_i(\mspd_j) + w_i\left(\tau_{t-1} - d_i\right)^+ \right)x_{ijt}.
\label{eq:EXT_IIdueobjective}
\end{eqnarray}

Since the \algdue~algorithm speeds up the jobs, we need to add the following regularity condition for the energy cost functions $\mathcal{E}_i(\mspd)$ in order to obtain performance bounds:
\begin{assumption}
	$\exists \pw\in\mathbb{N}^+$, such that $\mathcal{E}_i(\gamma\mspd_i) \leq \gamma^{\pw-1}\mathcal{E}_i(\mspd_i), ~~\forall \gamma\geq 1$.
	\label{cd:speed_up}
\end{assumption}

\begin{theorem}
	The \algdue~algorithm with $\gamma = \frac{(1+\epsilon)}{\alpha(1-\alpha)}$ and $\alpha = \frac{1}{2}$, is a $4^{\pw}(1+\epsilon)^{\pw-1}(1 + \delta)^{\pw-1}$-approximation algorithm for the \getr{prec}~problem, for all non-negative energy cost functions $\mathcal{E}_i(\spd)$ that satisfy Assumption \ref{cd:speed_up}.
	\label{th:EXT_IIdue_prec}
\end{theorem}
\begin{proof}
	As before, all the completion time related bounds (\ref{eq:IIdue_c_alpha}) and (\ref{eq:IIdue_ctilde}) remain valid, so only a bound analogous to (\ref{eq:IIdue_e_term}) is needed. From Assumption \ref{cd:speed_up} it follows that,
	\begin{eqnarray}
	\mathcal{E}_i(\tilde{\spd}_i^\alpha) \leq \gamma^{\pw-1} \mathcal{E}_i(\spd_i^\alpha) \leq \frac{(1+\epsilon)^{\pw-1}}{\alpha^{\pw}(1-\alpha)^{\pw}} \sum_{j=1}^m\sum_{t=1}^T \mathcal{E}_i(\mspd_j)\xb_{ijt}.
	\label{eq:EXT_IIdue_e_bound}
	\end{eqnarray}
	
	Thus, from (\ref{eq:IIdue_ctilde}) it follows that,
	\begin{eqnarray}
	\sum_{i=1}^n \mathcal{E}_i(\tilde{\spd}_i^\alpha) + \sum_{i=1}^n w_i \left(\tilde{C}_i^{\alpha} - d_i\right)^{+} \leq \frac{(1+\epsilon)^{\pw-1}}{\alpha^{\pw}(1-\alpha)^{\pw}} \left[\sum_{i=1}^n\sum_{j=1}^m\sum_{t=1}^T \mathcal{E}_i(\mspd_j) x_{ijt}^* + \sum_{i=1}^n w_i \left(C_i^* - d_i\right)^{+}\right]\nonumber. 
	\end{eqnarray}
	
	Since we are rounding speeds up, equation (\ref{eq:IIdue_rounding}) remains valid and thus taking $\alpha = \frac{1}{2}$ completes the proof.
\end{proof}

\subsection{Continuous Speeds}
\label{sec:extension_cspeeds}

As commented previously, our algorithms are also applicable for the case when a continuous set of speeds is possible. In this case we modify the \algI~and \algdue~algorithms, eliminating the rounding step required at the end of each algorithm.

When the operating range of the machine is given, i.e. the speed limits $\mspd_{\min}$ and $\mspd_{\max}$, since our IP requires a speed index, we need to quantize the set $[\mspd_{\min}, \mspd_{\max}]$ in $m$ different speeds. We can do this by setting $\mspd_1 = \mspd_{\min}$, and as before we define speed $\mspd_j = (1+\delta)\mspd_{j-1}$, for some $\delta>0$, making sure that $\mspd_m \geq \mspd_{\max}$ in order to cover the whole operating range. Just by rounding as described in Section \ref{sec:extension_wC} for the weighted completion time setting and rounding up for the weighted tardiness setting we can prove the following lemma:
\begin{lemma}
The optimal solution for the IP (\ref{eq:II_objective})-(\ref{eq:II_ctr5}) is at most $(1+\delta)$ times the optimal solution of the energy aware problem in the weighted completion time and continuous speed setting, and the optimal solution for the IP (\ref{eq:IIdue_objective}), (\ref{eq:II_ctr1})-(\ref{eq:II_ctr5}) is at most $(1+\delta)^{\pw-1}$ times the optimal solution of the energy aware problem in the weighted tardiness and continuous speed setting.
\label{lm:EXT_contspeed}
\end{lemma}
The proof is similar to Lemma \ref{lm:speeddown_bound} for the weighted completion time and similar to equation (\ref{eq:IIdue_rounding}) for the weighted tardiness setting.

Since there is no additional rounding at the end of the algorithm, using Lemma \ref{lm:EXT_contspeed} we get the same approximation ratios as in Theorems \ref{th:EXT_II_prec}, \ref{th:EXT_II_precr}, and \ref{th:EXT_IIdue_prec}.

When the operating range of the machine is not given, and we are interested in determining a set $\vk{S}$ that covers the optimal speeds from the continuous case, we need the following additional regularity condition on the energy cost functions: $\exists\xi < \infty$ such that $\mathcal{E}_j(\spd_i)$ is increasing $\forall \spd_i \geq \xi$. It is easy to prove that this is a necessary and sufficient conditions for the problem to be well defined, and thus we can compute $\mspd_{\min}$ and $\mspd_{\max}$ such that the optimal speeds $\spd_i^*\in[\mspd_{\min}, \mspd_{\max}]$, for all $i$. Then we can apply the same procedure as before to quantize and build the set of speeds, and proceed to compute an approximate solution.

\section{Conclusion}
\label{sec:conclusion}
In this work we described new techniques for developing constant approximation algorithms for energy aware scheduling problems with very general job-dependent energy cost functions, that work on both discrete and continuous speed sets. Furthermore, we present the first algorithm, to the best of our knowledge, that tackles the energy aware weighted tardiness setting, even in the presence of arbitrary precedence constraints.

We believe that our methodology, which extends the idea of
$\alpha$-points to the energy aware setting by developing the
$\alpha$-speeds concept, should have many more applications. We suspect that, 
via techniques such as using randomly chosen values of $\alpha$ or using 
different $\alpha$ values for different jobs, we could obtain tighter bounds,
and also that these techniques could be extended to other settings, such as 
multiple parallel machines among others.

\bibliographystyle{acm}
\bibliography{bibliography,bibliography_extras}

\end{document}

%% file: rax_hsetup.tex
\newcommand{\newhead}[2]{\fancyhead[L]{#1} \fancyhead[R]{#2}}